\documentclass[letterpaper,11pt,dvipsnames]{article}

\usepackage[utf8]{inputenc}
\usepackage[greek,english]{babel}
\usepackage[T1]{fontenc}
\usepackage[lmargin=1in,rmargin=1in,tmargin=1in,bmargin=1in]{geometry}
\usepackage{amsmath}
\usepackage{amsthm}
\usepackage{thmtools, thm-restate}
\usepackage{amssymb}
\usepackage{xspace}
\usepackage{url}
\usepackage{hyperref}
\usepackage{graphicx}
\usepackage{wrapfig}

\usepackage{tikz}
\usetikzlibrary{arrows.meta,calc,backgrounds,angles,quotes}

\usepackage[nospace,nocompress]{cite}

\makeatletter
\let\origsection\section
\renewcommand\section{\@ifstar{\starsection}{\nostarsection}}

\newcommand\nostarsection[1]
{\sectionprelude\origsection{#1}\sectionpostlude}

\newcommand\starsection[1]
{\sectionprelude\origsection*{#1}\sectionpostlude}

\newcommand\sectionprelude{%
  \vspace{-2pt}
}

\newcommand\sectionpostlude{%
  \vspace{-3pt}
}

\let\origsubsection\subsection
\renewcommand\subsection{\@ifstar{\starsubsection}{\nostarsubsection}}

\newcommand\nostarsubsection[1]
{\subsectionprelude\origsubsection{#1}\subsectionpostlude}

\newcommand\starsubsection[1]
{\subsectionprelude\origsubsection*{#1}\subsectionpostlude}

\newcommand\subsectionprelude{%
  \vspace{-2pt}
}

\newcommand\subsectionpostlude{%
  \vspace{-3pt}
}

\let\origparagraph\paragraph

\renewcommand\paragraph[1]
{\paragraphprelude\origparagraph{#1}\paragraphpostlude}

\newcommand\paragraphprelude{%
  \vspace{-2pt}
}

\newcommand\paragraphpostlude{%
  \vspace{-3pt}
}

\makeatother

\newtheorem{definition}{Definition}
\newtheorem{proposition}{Proposition}
\newtheorem{lemma}{Lemma}

\newcommand{\Kr}{\textsc{Kr}\xspace}

\title{Byzantine-Tolerant Machine Learning}

\author{Peva Blanchard \and El Mahdi El Mhamdi
\and Rachid Guerraoui \and Julien Stainer}
\date{École Polytechnique Fédérale de Lausanne\\\texttt{first.last@epfl.ch}}

\begin{document}
\nocite{*}
\pagenumbering{gobble}

\maketitle

\begin{abstract}
 The growth of data, the need for scalability and the complexity of models used in modern machine learning calls for distributed implementations. 
 Yet, as of today, distributed machine learning frameworks have largely ignored the possibility of 
 arbitrary (i.e., Byzantine) failures.
 In this paper, we study the robustness to Byzantine failures at the fundamental level of stochastic gradient descent (SGD), 
 the heart of most machine learning algorithms.
 Assuming a set of $n$ workers, up to $f$ of them being Byzantine,
 we ask how robust can SGD be, 
 without limiting the dimension, nor the size of the parameter space.

 We first show that no gradient descent update rule based on a linear combination of the vectors proposed by the workers  
 (i.e, current approaches) tolerates a single Byzantine failure.
 We then formulate a resilience property of the update rule capturing 
 the basic requirements to guarantee convergence despite $f$   Byzantine workers. 
 We finally propose \emph{Krum}, an update rule that satisfies the resilience property aforementioned. 
 For a $d$-dimensional learning problem, the time complexity of Krum is $O(n^2 \cdot (d + \log n))$.
 
\end{abstract}

\vspace{2cm}

\begin{center}


\end{center}

\vspace{2cm}

\bigskip
\bigskip


\clearpage
\pagenumbering{arabic}
\setcounter{page}{1}

\section{Introduction}

Machine learning has received a lot of attention over the past few years. 
Its applications range all the way from images classification, financial trend prediction, disease diagnosis, to gaming and driving~\cite{jordan2015machine}. 
Most major companies are currently investing in machine learning technologies to support their businesses~\cite{forbes2017}.
Roughly speaking, machine learning consists in giving a computer the ability to improve the way it solves a problem with the quantity and quality of information it can use~\cite{samuel1959some}.
In short, the computer has a list of internal parameters, called the \emph{parameter vector},  
which allows the computer to formulate answers to several questions such as, ``is there a cat on this picture?''.
According to how many correct and incorrect answers are provided, a specific error cost is associated with the parameter vector.
\textit{Learning} is the process of updating this parameter vector in order to minimize the cost.

The increasing amount of data involved~\cite{dean2012large} 
as well as the growing complexity of models~\cite{srivastava2015training} 
has led to learning schemes that require a lot of computational resources.
As a consequence, most industry-grade machine-learning implementations are now distributed~\cite{abadi2016tensorflow}.
For example, as of 2012, Google reportedly used 16.000 processors to train an image classifier~\cite{markoff2012many}.
However, distributing a computation over several machines 
induces a higher risk of failures, including crashes and computation errors. 
In the worst case, the system may undergo \emph{Byzantine} failures~\cite{lamport1982Byzantine}, i.e., 
completely arbitrary behaviors of some of the machines involved.
In practice, such failures may be due to stalled processes, or biases in the way the data samples are distributed among the processes.

A classical approach to mask failures in distributed systems is to use a state machine replication protocol~\cite{schneider1990implementing}, 
which requires however state transitions to be applied by all processes. 
In the case of distributed machine learning, this constraint can be seen in two ways: 
either \emph{(a)} the processes agree on a sample of data based on which they update their local parameter vectors, 
or \emph{(b)} they agree on how the parameter vector should be updated. 
In case \emph{(a)}, the sample of data has to be transmitted to each process, which then has to perform a heavyweight computation 
to update its local parameter vector. 
This entails communication and computational costs that defeat the entire purpose of 
distributing the work.
In case \emph{(b)}, the processes have no way to check if the chosen update for the parameter vector has indeed 
been computed correctly on real data (a Byzantine process could have proposed the update). 
Byzantine failures may easily prevent the convergence of the learning algorithm. 
Neither of these solutions is satisfactory in a realistic distributed machine learning setting.

In fact, most learning algorithms today rely on a core component, namely \emph{stochastic gradient descent} (SGD)~\cite{bottou2010large, haykin2009neural}, whether for training neural networks~\cite{haykin2009neural}, regression~\cite{zhang2004solving}, 
matrix factorization~\cite{gemulla2011large} or support vector machines~\cite{zhang2004solving}. 
In all those cases, a cost function – depending on the parameter vector – is minimized based on stochastic estimates of its gradient. 
Distributed implementations of SGD~\cite{zhang2015deep} typically take the following form: 
a single parameter server is in charge of updating the parameter vector, while worker processes perform the actual update estimation, based on the share of data they have access to.
More specifically, the parameter server executes  synchronous rounds, 
during each of which, the parameter vector is broadcast to the workers.
In turn, each worker computes an estimate of the update to apply (an estimate of the \emph{gradient}), 
and the parameter server aggregates their results to finally update the parameter vector.
Today, this aggregation is typically implemented through averaging \cite{polyak1992acceleration}, 
or variants of it \cite{zhang2015deep, lian2015asynchronous, tsitsiklis1986distributed}.

The question we address in this paper is  
how a distributed SGD can be devised to tolerate $f$ Byzantine processes among the $n$ workers.

\paragraph{Contributions.}
We first show in this paper that no linear combination (current approaches) of the updates proposed by the workers 
can tolerate a \emph{single} Byzantine worker. 
Basically, 
the Byzantine worker can force the parameter server to choose any arbitrary vector, even one that is too large in amplitude or too far in direction from the other vectors. Clearly, the Byzantine worker can prevent any classic averaging-based approach to converge. Choosing the appropriate update from the vectors proposed by the workers turns out to be challenging.
A non-linear, \textit{distance-based} choice function, that chooses, among the proposed vectors, 
the vector ``closest to everyone else'' 
(for example by taking the vector that minimizes the sum of the distances to every other vector), might look appealing. 
Yet, such a distance-based choice tolerates only a single Byzantine worker. 
Two Byzantine workers can collude, one helping the other to be selected, by moving the barycenter of all the vectors farther from the ``correct area''.

We formulate a Byzantine resilience property
capturing sufficient conditions for the parameter server's choice to tolerate $f$ Byzantine workers. 
Essentially, to guarantee that the cost will decrease despite Byzantine workers, we require the parameter server's choice 
\emph{(a)} to point, on average, to the same direction as the gradient and 
\emph{(b)} to have statistical moments (up to the fourth moment) 
bounded above by a homogeneous polynomial in the moments of a correct estimator of the gradient.
One way to ensure such a resilience property is to consider a \textit{majority-based} approach, looking  at every subset of $n-f$ vectors, and considering the subset with the smallest diameter. 
While this approach is more robust to Byzantine workers that propose vectors far from the correct area, 
its exponential computational cost is prohibitive.
Interestingly, combining the intuitions of the \textit{majority-based} and \textit{distance-based} methods, 
we can choose the vector that is somehow the closest to its $n-f$ neighbors. Namely, the one that minimizes a distance-based criteria, but only within its $n-f$ neighbors.
This is the main idea behind our choice function we call \emph{Krum}\footnote{{Krum, in Greek \textgreek{Κρούμος},
was a Bulgarian Khan of the end of the eighth century, who undertook offensive attacks against the Byzantine empire. Bulgaria doubled in size during his reign.}}. 
Assuming $2f+2 < n$, we show (using techniques from multi-dimensional stochastic calculus)
that our Krum function satisfies the resilience property aforementioned and the corresponding machine learning scheme converges. 
An important advantage of the Krum function is that it requires $O(n^2 \cdot (d + \log n))$ local 
computation time, 
where $d$ is the dimension of the parameter vector.
(In modern machine learning, the dimension $d$ of the parameter vector may take values in the hundreds of  billions~\cite{trask2015modeling}.)
For simplicity of presentation, we first introduce a version of the Krum function that selects only one vector.
Then we discuss how this method can be iterated to leverage the contribution of more than one single correct worker.

\paragraph{Paper Organization.} 
Section~\ref{sec:model} recalls the classical model of distributed SGD. 
Section~\ref{sec:byzresilience} proves that linear combinations (solutions used today) are not resilient even to a single Byzantine worker, 
then introduces the new concept of $(\alpha,f)$-Byzantine resilience. 
In Section~\ref{sec:krum}, we introduce the Krum function, compute its computational cost and prove its $(\alpha,f)$-Byzantine resilience. 
In Section~\ref{sec:cvanalysis} we analyze the convergence of a distributed SGD using our Krum function. In Section~\ref{sec:multikrum} we discuss how Krum can be iterated to leverage the contribution of more workers.
Finally, we discuss related work and open problems in Section~\ref{sec:conclusion}.


\section{Model}
\label{sec:model}


We consider a general distributed system consisting of a parameter server\footnote{The parameter server is assumed to be reliable. 
Classical techniques of state-machine replication can be used to avoid this single point of failure.}~\cite{abadi2016tensorflow}, and $n$ workers, $f$ of them possibly Byzantine (behaving arbitrarily). 
Computation is divided into (infinitely many) synchronous rounds.
During round $t$, the parameter server broadcasts its parameter vector $x_t \in \mathbb{R}^d$ to all the workers. 
Each correct worker $p$ computes an estimate $V_p^t = G(x_t,\xi_p^t)$ of the gradient $\nabla Q(x_t)$ of the cost function $Q$, 
where $\xi_p^t$ is a random variable representing, e.g., the sample drawn from the dataset. 
A Byzantine worker $b$ proposes a vector $V_b^t$ which can be arbitrary (see Figure~\ref{fig:gradients}).
\begin{wrapfigure}{r}{.5\linewidth}
\centering
\begin{tikzpicture}
\draw[domain=-5:.5] plot (\x,{pow(\x,2)/10});
\draw[-{Stealth[scale=1]},blue,very thick] (-3,{pow(-3,2)/10}) -- (-1,{pow(-3,2)/10+2*(-3)/10*(-1-(-3))});
\draw[-{Stealth[scale=1]},black,dashed,very thick] (-3,{pow(-3,2)/10}) -- (-1+0.2,{pow(-3,2)/10+2*(-3)/10*(-1-(-3))+0.2});
\draw[-{Stealth[scale=1]},black,dashed,very thick] (-3,{pow(-3,2)/10}) -- (-1-0.1,{pow(-3,2)/10+2*(-3)/10*(-1-(-3))-0.2});
\draw[-{Stealth[scale=1]},black,dashed,very thick] (-3,{pow(-3,2)/10}) -- (-1-0.3,{pow(-3,2)/10+2*(-3)/10*(-1-(-3))-0.3});
\draw[-{Stealth[scale=1]},red,dotted,very thick] (-3,{pow(-3,2)/10}) -- (-5,0);
\fill[black] (-3,{pow(-3,2)/10}) circle (1pt);
\end{tikzpicture}
\caption{The gradient estimates computed by correct workers (black dashed arrows) are distributed around the actual gradient (blue solid arrow) of the cost function (thin black curve). A Byzantine worker can propose an arbitrary vector (red dotted arrow).}
\label{fig:gradients}
\end{wrapfigure}
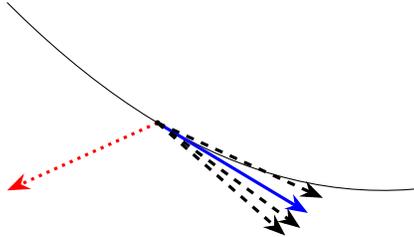
Note that, since the communication is synchronous, if the parameter server does not receive a vector value $V_b^t$ from a given Byzantine worker $b$, then the parameter server acts as if it had received the default value $V_b^t = 0$ instead.

The parameter server computes a vector $F(V_1^t,\dots,V_n^t)$ by applying a deterministic function $F$ to the vectors received.
We refer to $F$ as the \emph{choice function} of the parameter server.
The parameter server updates the parameter vector using the following SGD equation
\begin{equation*}
  x_{t+1} = x_t - \gamma_t \cdot F(V_1^t,\dots,V_n^t).
\end{equation*}


In this paper, we assume that the correct (non-Byzantine) workers compute unbiased estimates of the gradient $\nabla Q(x_t)$. More precisely, 
in every round $t$, the vectors  $V^t_i$'s proposed by the correct  workers are independent identically distributed random vectors, $V^t_i \sim G(x_t, \xi_i^t)$ with $\mathbb{E}G(x_t,\xi_i^t) = \nabla Q(x_t)$. This can be achieved by ensuring that each sample of data used for computing the gradient is drawn uniformly and independently, as classically assumed in the literature of machine learning~\cite{bottou1998online}.

The Byzantine workers have full knowledge of the system, including the choice function $F$, the vectors proposed by the other workers and can collaborate with each other~\cite{lynch1996distributed}.




\section{Byzantine Resilience}
\label{sec:byzresilience}

In most SGD-based learning algorithms used today~\cite{bottou2010large, haykin2009neural, zhang2004solving, gemulla2011large}, the choice function  consists in computing the average of the input vectors.
Lemma~\ref{lem:average} below states that no linear combination of the vectors can tolerate  a single Byzantine worker.
In particular, averaging is not robust to Byzantine failures.

\begin{lemma}
\label{lem:average}
  Consider a choice function $F_{lin}$ of the form
  \begin{equation*}
    F_{lin}(V_1,\dots,V_n) =  \sum_{i=1}^n \lambda_i \cdot V_i.
  \end{equation*}
  where the $\lambda_i$'s are non-zero scalars.
  Let $U$ be any vector in $\mathbb{R}^d$. A single Byzantine worker can make $F$ always select $U$.
  In particular, a single Byzantine worker can prevent convergence.
\end{lemma}

\begin{proof}
 If the Byzantine worker
 proposes  vector
   $V_n = \frac{1}{\lambda_n} \cdot U - \sum_{i=1}^{n-1} \frac{\lambda_i}{\lambda_n} V_i$,
 then $F = U$. 
 Note that the parameter server could cancel the effects of the Byzantine behavior by setting,
 for example, $\lambda_n$  to 0, but this  requires means to detect which worker is Byzantine.
\end{proof}

In the following, we define basic requirements on an appropriate robust choice function.
Intuitively, the choice function should output a vector $F$ that is not too far from the ``real'' gradient $g$, 
more precisely, the vector that points to the steepest direction of the cost function being optimized.
This is expressed as a lower bound (condition \emph{(i)}) on the scalar product of the (expected) vector $F$ and $g$.
Figure~\ref{fig:choicefunctioncontrol} illustrates the situation geometrically.
If $\mathbb{E}F$ belongs to the ball centered at $g$ with radius $r$, 
then the scalar product is bounded below by a term involving  $\sin\alpha = r/\lVert g \rVert$.

Condition \emph{(ii)} is more technical, and states that the moments of $F$ should be controlled by the moments of the (correct) gradient estimator $G$.
The bounds on the moments of $G$ are classically used to control the effects of the discrete nature of the SGD dynamics~\cite{bottou1998online}.
Condition \emph{(ii)} allows to transfer this control to the choice function.

\begin{figure}
\centering
\begin{tikzpicture}[scale=1.4]
\coordinate (v1) at (0,0);
\coordinate (v2) at (5,0);
\coordinate (v3) at (5,1);
\draw (v1) -- (v2) -- node[right] {$r$} (v3) pic[draw,"$\alpha$",angle radius=2cm, angle eccentricity=.9] {angle=v2--v1--v3};
\draw[-{Stealth[scale=1.5]}] (v1) -- node[above] {$g$} (v3);
\draw (v3) circle (1);
\draw ($(v2)+(-.1,0)$) -- ($(v2)+(-.1,.1)$) -- ($(v2)+(-0,.1)$);
\end{tikzpicture}
\caption{If $\left\lVert \mathbb{E}F - g \right\rVert \le r$ then $\langle \mathbb{E}F,g \rangle$ is bounded below by $(1-\sin\alpha) \lVert g \rVert^2$ where $\sin \alpha = r/\lVert g \rVert$.}
\label{fig:choicefunctioncontrol}
\end{figure}
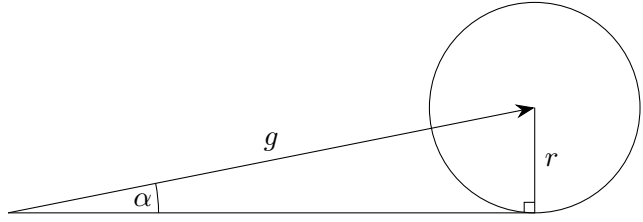

\begin{definition}[$(\alpha,f)$-Byzantine Resilience]
 Let $0 \le \alpha < \pi/2$ be any angular value, and any integer $0 \le f \le n$.
 Let $V_1,\dots,V_n$ be any independent identically distributed random vectors in $\mathbb{R}^d$, $V_i \sim G$, with $\mathbb{E}G = g$.
 Let $B_1,\dots,B_f$ be any random vectors in $\mathbb{R}^d$, possibly dependent on the $V_i$'s.
 Choice function $F$ is said to be $(\alpha,f)$-Byzantine resilient if, for any $1 \leq j_1 < \dots < j_f \leq n$, the vector
 \begin{equation*}
   F = F(V_1,\dots,\underbrace{B_1}_{j_1},\dots,\underbrace{B_f}_{j_f},\dots,V_n)
 \end{equation*}
 satisfies
    \emph{(i)} $\langle \mathbb{E}F, g \rangle \ge (1- \sin\alpha) \cdot \lVert g \rVert^2 > 0$ and
    \emph{(ii)} for $r = 2,3,4$,  $\mathbb{E}\left\lVert F \right\rVert^r$   is bounded above by 
    a linear combination of terms $\mathbb{E}\left\lVert G \right\rVert^{r_1} \dots \mathbb{E} \left\lVert G \right\rVert^{r_{n-1}}$
    with $r_1 + \dots + r_{n-1} = r$.
    \label{def:f-byz-resilience}
\end{definition}

\section{The Krum Function}
\label{sec:krum}

We now introduce \emph{Krum}, our choice function, which, we show, 
satisfies the $(\alpha, f)$-Byzantine resilience condition.
The barycentric choice function
  $F_{bary} = \frac{1}{n} \sum_{i=1}^n V_i$
can be defined as the vector in $\mathbb{R}^d$
that minimizes the sum of squared distances to the $V_i$'s
  $\sum_{i=1}^n \left\lVert F_{bary} - V_i \right\rVert^2$.
Lemma~\ref{lem:average}, however, states that this approach does not tolerate even a single Byzantine failure.
One could try to define the choice function 
in order to select, \emph{among} the $V_i$'s, the vector $U \in \{V_1,\dots,V_n\}$ that minimizes the sum $\sum_i \left\lVert U - V_i \right\rVert^2$.
Intuitively, vector $U$ would be close to every proposed vector, 
including the correct ones, and thus would be close to the ``real'' gradient.
However, all  Byzantine workers but one
may propose vectors that are large enough to move the total barycenter far away from the correct vectors,
while the remaining Byzantine worker proposes this barycenter. 
Since the barycenter always minimizes the sum of squared distance, 
this last Byzantine worker is certain to have its vector chosen by the parameter server. 
This situation is depicted in Figure~\ref{fig:byz-taking-over}.
In other words, since this choice function takes into account all the vectors, 
including the very remote ones, the Byzantine workers can collude to force the choice of the parameter server.

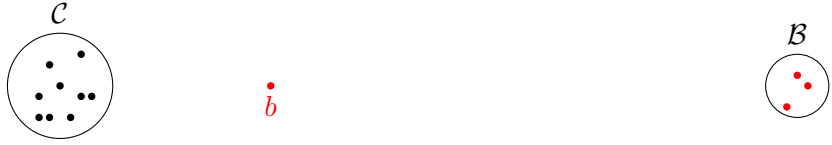
\begin{figure}[ht!]
\centering
\begin{tikzpicture}[scale=1.4]
\coordinate (v1) at (-.2,-.1);
\coordinate (v2) at (-.1,.2);
\coordinate (v3) at (.3,-.1);
\coordinate (v4) at (0,0);
\coordinate (v5) at (-.1,-.3);
\coordinate (v6) at (-.2,-.3);
\coordinate (v7) at (.1,-.3);
\coordinate (v8) at (.2,-.1);
\coordinate (v9) at (.2,.3);
\foreach \point in {v1,v2,v3,v4,v5,v6,v7,v8,v9}
    \fill[black] (\point) circle (1pt);

\draw (0,0) circle (.5);
\node[above] at (0,.5) {${\cal C}$};

\coordinate (b1) at (7,.1);
\coordinate (b2) at (6.9,-.2);
\coordinate (b3) at (7.1,0);
\foreach \point in {b1,b2,b3}
    \fill[red] (\point) circle (1pt);

\draw (7,0) circle (.3);
\node[above] at (7,.3) {${\cal B}$};

\coordinate (b4) at (2,0);
\fill[red] (b4) circle (1pt) node[below] {$b$};
\end{tikzpicture}
\caption{Selecting the vector that minimizes the sum of the squared distances to other vectors does not prevent arbitrary vectors proposed by Byzantine workers from being selected if $f\geq 2$. 
If the gradients computed by the correct workers lie in area ${\cal C}$, the Byzantine workers can collude to propose up to $f-1$ vectors in an arbitrarily remote area ${\cal B}$, thus allowing another Byzantine vector $b$, close to the barycenter of proposed vectors, to be selected.}
\label{fig:byz-taking-over}
\end{figure}

Our approach to circumvent this issue is to preclude the vectors that are too far away.
More precisely, we define our \emph{Krum} choice function  $\Kr(V_1,\dots,V_n)$ as follows. 
For any $i \ne j$, we denote by $i \rightarrow j$ the fact that $V_j$ belongs to the $n-f-2$ closest vectors to $V_i$.
Then, we define for each worker $i$, the \emph{score}
  $s(i) = \sum_{i \rightarrow j} \left\lVert V_i - V_j \right\rVert^2$
where the sum runs over the $n-f-2$ closest vectors to $V_i$.
Finally, $\Kr(V_1,\dots,V_n) = V_{i_*}$ where $i_*$ refers to 
the
worker minimizing the score, $s(i_*) \le s(i)$ for all $i$.\footnote{If two or more workers have the minimal score, we choose the vector of the worker with the smallest identifier.}

\begin{lemma}
The time complexity of the Krum Function $\Kr(V_1,\dots,V_n)$, where $V_1,\dots,V_n$ are $d$-dimensional vectors, is $O(n^2 \cdot (d + \log n))$
\end{lemma}

\begin{proof}
For each $V_i$, the parameter server computes the $n$ squared distances $\left\lVert V_i - V_j \right\rVert^2$ (time $O(n \cdot d)$). 
Then the parameter server sorts these distances (time $O(n \cdot \log n)$) and sums the first $n-f-1$ values (time $O(n\cdot d)$).
Thus, computing the score of all the $V_i$'s takes $O(n^2 \cdot (d + \log n))$. 
An additional term $O(n)$ is required to find the minimum score, but is negligible relatively to $O(n^2 \cdot (d + \log n))$.
\end{proof}


Proposition~\ref{prop:krumresilient} below states that, if $2f+2 < n$ and the gradient estimator is accurate enough,
(its standard deviation is relatively small compared to the norm of the gradient), 
then the Krum function is $(\alpha,f)$-Byzantine-resilient,
where  angle $\alpha$ depends on the ratio of the deviation over the gradient.
When the Krum function selects a correct vector (i.e., a vector proposed by a correct worker), the proof of this fact is relatively easy, 
since the probability distribution of this correct vector is that of the gradient estimator $G$.
The core difficulty occurs when the Krum function selects a Byzantine vector (i.e., a vector proposed by a Byzantine worker), 
because the distribution of this vector is completely arbitrary, and may even depend on the correct vectors.
In a very general sense, this part of our proof is reminiscent of the median technique: 
the median of $n > 2f$ scalar values is always bounded below and above by values proposed by correct workers.
Extending this observation to our multi-dimensional is not trivial.
To do so, we notice that the chosen Byzantine vector $B_k$ has a score not greater than any score of a correct worker.
This allows us to derive an upper bound  on the distance between $B_k$ and the real gradient.
This upper bound involves a sum of distances from correct to correct neighbor vectors, 
and distances from correct to Byzantine neighbor vectors.
As explained above, the first term is relatively easy to control. 
For the second term, 
we observe that a correct vector $V_i$ has $n-f-2$ neighbors (the $n-f-2$ closest vectors to $V_i$), 
and $f+1$ non-neighbors. 
In particular, the distance from any (possibly Byzantine) neighbor $V_j$ to $V_i$ 
is bounded above by a correct to correct vector distance. 
In other words, we manage to control the distance between the chosen Byzantine vector 
and the real gradient by an upper bound involving only distances between vectors proposed by correct workers.

\begin{proposition}
  Let $V_1,\dots,V_n$ be any independent and identically distributed random $d$-dimensional vectors s.t $V_i \sim G$, with $\mathbb{E}G = g$ 
  and $\mathbb{E}\left\lVert G - g \right\rVert^2 = d\sigma^2$. 
  Let $B_1,\dots,B_f$ be any $f$ random vectors, possibly dependent on the $V_i$'s.
  If $2f + 2 < n$  and $\eta(n,f)\sqrt{d}\cdot \sigma < \lVert g \rVert$,
  where
  \begin{equation*}
    \eta(n,f) \underset{def}{=} \sqrt{2\left(n-f + \frac{f\cdot (n-f-2) + f^2\cdot(n-f-1)}{n-2f-2}\right)}
    = \left\{\begin{array}{cc}
     O(n) &  \text{if } f = O(n)\\
     O(\sqrt{n}) & \text{if } f = O(1) 
    \end{array}
    \right.,
  \end{equation*}
  then the Krum function $\Kr$ is $(\alpha,f)$-Byzantine resilient where $0 \le \alpha < \pi/2$ is defined by
  \begin{equation*}
    \sin \alpha = \frac{ \eta(n,f) \cdot \sqrt{d} \cdot \sigma }{\lVert g \rVert}.
  \end{equation*}
  \label{prop:krumresilient}
\end{proposition}
%
%
\vspace{-15pt}
The condition on the norm of the gradient, $\eta(n,f) \cdot \sqrt{d} \cdot \sigma < \lVert g \rVert$, can be satisfied, to a certain extent, 
by having the (correct) workers computing their gradient estimates on mini-batches~\cite{bottou1998online}. 
Indeed, averaging the gradient estimates over a mini-batch divides the deviation $\sigma$ by the squared root of the size of the mini-batch.

\begin{proof}
  Without loss of generality, we assume that the Byzantine vectors $B_1,\dots,B_f$ occupy the last $f$ positions in the list of arguments of $\Kr$,
  i.e., $\Kr = \Kr(V_1,\dots,V_{n-f},B_1,\dots,B_f)$. 
  An index is \emph{correct} if it refers to a vector among $V_1,\dots,V_{n-f}$. 
  An index is \emph{Byzantine} if it refers to a vector among $B_1,\dots,B_f$.
  For each index (correct or Byzantine) $i$, we denote by $\delta_c(i)$ (resp. $\delta_b(i)$) 
  the number of correct (resp. Byzantine) indices $j$ such that $i \rightarrow j$.
  We have
  \begin{align*}
    \delta_c(i) + &\delta_b(i) = n-f-2 \\
      n-2f-2 \le &\delta_c(i) \le n-f-2 \\
          &\delta_b(i) \le f.
  \end{align*}
  We focus first on the condition \emph{(i)} of $(\alpha,f)$-Byzantine resilience.
  We determine an upper bound on the squared distance $\lVert \mathbb{E}\Kr - g \rVert^2$.
  Note that, for any correct $j$, $\mathbb{E}V_j = g$. We denote by $i_*$ the index of the vector chosen by the Krum function.
  
  \vspace{-10pt}
  
  \begin{align*}
     \left\lVert \mathbb{E}\Kr - g \right\rVert^2 &\le \left\lVert \mathbb{E} \left( \Kr - \frac{1}{\delta_c(i_*)} \sum_{i_* \rightarrow \text{ correct } j} V_j \right) \right\rVert^2 \\
         &\le \mathbb{E} \left\lVert \Kr - \frac{1}{\delta_c(i_*)} \sum_{i_* \rightarrow \text{ correct } j} V_j \right\rVert^2 ~~ \text{(Jensen inequality)} \\
         &\le \sum_{\text{correct } i} \mathbb{E} \left\lVert V_i - \frac{1}{\delta_c(i)} \sum_{i \rightarrow \text{ correct } j} V_j \right\rVert^2 \mathbb{I}(i_* = i) \\
         &+ \sum_{\text{byz } k} \mathbb{E}\left\lVert B_k - \frac{1}{\delta_c(k)} \sum_{k \rightarrow \text{ correct } j} V_j \right\rVert^2 \mathbb{I}(i_* = k)
  \end{align*}
  where $\mathbb{I}$ denotes the indicator function\footnote{$\mathbb{I}(P)$ equals $1$ if the predicate $P$ is true, and $0$ otherwise.}.
  We examine the case $i_* = i$ for some correct index $i$. 
  \begin{align*}
    \left\lVert V_i - \frac{1}{\delta_c(i)} \sum_{i \rightarrow \text{ correct } j} V_j \right\rVert^2 &= \left\lVert \frac{1}{\delta_c(i)} \sum_{i \rightarrow \text{ correct } j} V_i - V_j \right\rVert^2 \\
      &\le \frac{1}{\delta_c(i)} \sum_{i \rightarrow \text{ correct } j} \left\lVert V_i - V_j \right\rVert^2 ~~ \text{(Jensen inequality)} \\
    \mathbb{E} \left\lVert V_i - \frac{1}{\delta_c(i)} \sum_{i \rightarrow \text{ correct } j} V_j \right\rVert^2  &\le \frac{1}{\delta_c(i)} \sum_{i \rightarrow \text{ correct } j} \mathbb{E}\left\lVert V_i - V_j \right\rVert^2 \\
      &\le 2d\sigma^2.
  \end{align*}
  We now examine the case $i_* = k$ for some Byzantine index $k$.
  The fact that $k$ minimizes the score implies that for all correct indices $i$
  \begin{equation*}
    \sum_{k \rightarrow \text{ correct } j} \left\lVert B_k - V_j \right\rVert^2 
      + \sum_{k \rightarrow \text{ byz } l} \left\lVert B_k - B_l \right\rVert^2 
      \le \sum_{i \rightarrow \text{ correct } j} \left\lVert V_i - V_j \right\rVert^2 
      + \sum_{i \rightarrow \text{ byz } l} \left\lVert V_i - B_l \right\rVert^2.
  \end{equation*}
  Then, for all correct indices $i$
  \begin{align*}
    \left\lVert B_k - \frac{1}{\delta_c(k)} \sum_{k \rightarrow \text{ correct } j} V_j \right\rVert^2
      &\le \frac{1}{\delta_c(k)} \sum_{k \rightarrow \text{ correct } j} \left\lVert B_k - V_j \right\rVert^2 \\
      &\le  \frac{1}{\delta_c(k)} \sum_{i \rightarrow \text{ correct } j} \left\lVert V_i - V_j \right\rVert^2 + \frac{1}{\delta_c(k)} \underbrace{\sum_{i \rightarrow \text{ byz } l} \left\lVert V_i - B_l \right\rVert^2 }_{D^2(i)}.
  \end{align*}
  We focus on the term $D^2(i)$. Each correct worker $i$ has $n-f-2$ neighbors, and $f+1$ non-neighbors. 
  Thus there exists a correct worker $\zeta(i)$ which is farther from $i$ than any of the neighbors of $i$. 
  In particular, for each Byzantine index $l$ such that $i \rightarrow l$, $\left\lVert V_i - B_l \right\rVert^2 \le \left\lVert V_i - V_{\zeta(i)}\right\rVert^2$. Whence
  \begin{align*}
    \left\lVert B_k - \frac{1}{\delta_c(k)} \sum_{k \rightarrow \text{ correct } j} V_j \right\rVert^2
      &\le \frac{1}{\delta_c(k)} \sum_{i \rightarrow \text{ correct } j} \left\lVert V_i - V_j \right\rVert^2 + \frac{\delta_b(i)}{\delta_c(k)} \left\lVert V_i - V_{\zeta(i)}\right\rVert^2 \\
    \mathbb{E} \left\lVert B_k - \frac{1}{\delta_c(k)} \sum_{k \rightarrow \text{ correct } j} V_j \right\rVert^2
       &\le \frac{\delta_c(i)}{\delta_c(k)} \cdot 2d\sigma^2 + \frac{\delta_b(i)}{\delta_c(k)} \sum_{\text{correct } j \ne i} \mathbb{E} \left\lVert V_i - V_j \right\rVert^2 \mathbb{I}(\zeta(i) = j) \\
       &\le \left(\frac{\delta_c(i)}{\delta_c(k)} \cdot + \frac{\delta_b(i)}{\delta_c(k)} (n-f-1)\right) 2d\sigma^2 \\
       &\le \left(\frac{n-f-2}{n-2f-2} + \frac{f}{n-2f-2} \cdot (n-f-1)\right) 2d\sigma^2.
  \end{align*}

  Putting everything back together, we obtain
  \begin{align*}
    \left\lVert \mathbb{E}\Kr - g \right\rVert^2 &\le (n-f) 2d \sigma^2 + f \cdot \left(\frac{n-f-2}{n-2f-2} + \frac{f}{n-2f-2} \cdot (n-f-1)\right) 2d \sigma^2 \\
       &\le \underbrace{2\left(n-f + \frac{f\cdot (n-f-2) + f^2\cdot(n-f-1)}{n-2f-2}\right)}_{\eta^2(n,f)} d\sigma^2.
  \end{align*}
  By assumption, $\eta(n,f) \sqrt{d} \sigma < \lVert g \rVert$, i.e., $\mathbb{E}\Kr$ belongs to a ball centered at $g$ with radius $\eta(n,f) \cdot \sqrt{d} \cdot \sigma$. This implies
  \begin{equation*}
    \langle \mathbb{E}\Kr, g \rangle \ge \left( \lVert g \rVert - \eta(n,f) \cdot \sqrt{d} \cdot \sigma  \right) \cdot \lVert g \rVert
       = (1 - \sin \alpha) \cdot \lVert g \rVert^2.
  \end{equation*}
  To sum up, condition~\emph{(i)} of the $(\alpha,f)$-Byzantine resilience property holds.
  We now focus on condition~\emph{(ii)}.
  \begin{align*}
    \mathbb{E} \lVert \Kr \rVert^r &= \sum_{\text{correct } i} \mathbb{E} \left\lVert V_i \right\rVert^r \mathbb{I}(i_* = i) 
                                      + \sum_{\text{byz } k} \mathbb{E} \left\lVert B_k \right\rVert^r \mathbb{I}(i_* = k)\\
            &\le (n-f) \mathbb{E} \left\lVert G \right\rVert^r + \sum_{\text{byz } k} \mathbb{E} \left\lVert B_k \right\rVert^r \mathbb{I}(i_* = k).
  \end{align*}
  Denoting by $C$ a generic constant, when $i_* = k$, we have for all correct indices $i$
  \begin{align*}
    \left\lVert B_k - \frac{1}{\delta_c(k)} \sum_{k \rightarrow \text{correct } j} V_j \right\rVert
      &\le \sqrt{\frac{1}{\delta_c(k)} \sum_{i \rightarrow \text{ correct } j} \left\lVert V_i - V_j \right\rVert^2 + \frac{\delta_b(i)}{\delta_c(k)} \left\lVert V_i - V_{\zeta(i)}\right\rVert^2} \\
      &\le C \cdot \left( \sqrt{\frac{1}{\delta_c(k)}} \cdot \sum_{i \rightarrow \text{correct } j} \left\lVert V_i - V_j\right\rVert +  \sqrt{\frac{\delta_b(i)}{\delta_c(k)}} \cdot \left\lVert V_i - V_{\zeta(i)}\right\rVert \right)\\
      &\le C \cdot \sum_{\text{correct } j} \left\lVert V_j \right\rVert ~~ \text{(triangular inequality)}.
  \end{align*}
  The second inequality comes from the equivalence of norms in finite dimension.
  Now
  \begin{align*}
    \left\lVert B_k \right\rVert  &\le \left\lVert B_k - \frac{1}{\delta_c(k)} \sum_{k \rightarrow \text{correct } j} V_j \right\rVert + \left\lVert \frac{1}{\delta_c(k)} \sum_{k \rightarrow \text{correct } j} V_j \right\rVert \\
        &\le C \cdot \sum_{\text{correct } j} \left\lVert V_j \right\rVert \\
    \left\lVert B_k \right\rVert^r &\le C \cdot \sum_{r_1 + \dots + r_{n-f} = r} \left\lVert V_1 \right\rVert^{r_1} \cdots \left\lVert V_{n-f} \right\rVert^{r_{n-f}}.
  \end{align*}
  Since the $V_i$'s are independent, we finally obtain that $\mathbb{E} \left\lVert \Kr \right\rVert^r$ is bounded above by a linear combination
  of terms of the form $\mathbb{E} \left\lVert V_1 \right\rVert^{r_1} \cdots \mathbb{E}\left\lVert V_{n-f} \right\rVert^{r_{n-f}} = \mathbb{E} \left\lVert G \right\rVert^{r_1} \cdots \mathbb{E}\left\lVert G \right\rVert^{r_{n-f}}$ with $r_1 + \dots + r_{n-f} = r$.
  This completes the proof of condition \emph{(ii)}.
\end{proof}

\section{Convergence Analysis}
\label{sec:cvanalysis}

In this section, we analyze the convergence of the SGD using our Krum function defined in Section~\ref{sec:krum}.
The SGD equation is expressed as follows
\begin{equation*}
  x_{t+1} = x_t - \gamma_t \cdot \Kr(V^t_1,\dots,V^t_n)
\end{equation*}
where at least $n-f$ vectors among the $V^t_i$'s are correct, while the other ones may be Byzantine.
For a correct index $i$, $V^t_i = G(x_t,\xi^t_i)$ where $G$ is the gradient estimator.
We define the \emph{local standard deviation} $\sigma(x)$ by
\begin{equation*}
  d \cdot \sigma^2(x) = \mathbb{E} \left\lVert G(x,\xi) - \nabla Q(x) \right\rVert^2.
\end{equation*}


The following proposition considers an (\emph{a priori}) non-convex cost function.
In the context of non-convex optimization, even in the centralized case, 
it is generally hopeless to aim at proving that the parameter vector $x_t$ tends to a local minimum.
Many criteria may be used instead. We follow~\cite{bottou1998online}, 
and we prove that the parameter vector $x_t$ almost surely 
reaches a ``flat'' region (where the norm of the gradient is small), 
in a sense explained below.

\begin{restatable}{proposition}{convergence}
  We assume that
  \emph{(i)} the cost function $Q$ is three times differentiable with continuous derivatives, and is non-negative, $Q(x) \geq 0$;
  \emph{(ii)} the learning rates satisfy
      $\sum_t \gamma_t = \infty$ and  $\sum_t \gamma_t^2 < \infty$;
  \emph{(iii)} the gradient estimator satisfies
      $\mathbb{E} G(x,\xi) = \nabla Q(x)$ and
      $\forall r \in \{2,\dots,4\},~ \mathbb{E} \lVert G(x,\xi) \rVert^r \le A_r + B_r \lVert x \rVert^r$
    for some constants $A_r,B_r$;
  \emph{(iv)} there exists a constant $0 \le \alpha < \pi/2$
      such that for all $x$
      \begin{equation*}
        \eta(n,f) \cdot \sqrt{d} \cdot \sigma(x) \le \lVert \nabla Q(x) \rVert \cdot \sin\alpha;
      \end{equation*}
  \emph{(v)} finally, beyond a certain horizon, $\lVert x \rVert^2 \geq D$, there exist $\epsilon > 0$ and $0 \le \beta < \pi/2 - \alpha$ such that
    \begin{align*}
      \left\lVert \nabla Q(x) \right\rVert &\ge \epsilon > 0 \\
      \frac{ \langle x, \nabla Q(x) \rangle}{\lVert x \rVert \cdot \lVert \nabla Q(x) \rVert} &\ge \cos \beta.
    \end{align*}
    Then the sequence of gradients $\nabla Q(x_t)$ converges almost surely to zero.
    \label{prop:convergence}
\end{restatable}

\begin{figure}
\centering
\begin{tikzpicture}[scale=1.4]
\coordinate (v4) at (-1,-1);
\coordinate (v5) at (2,2);
\coordinate (v1) at (0,0);
\coordinate (v2) at (5,0);
\coordinate (v3) at (5,1);
\draw (v4) -- (v5);
\draw (v1) -- (v2) -- node[right] {$\eta\sqrt{d}\sigma$} (v3) pic[draw,"$\alpha$",angle radius=2cm, angle eccentricity=.9] {angle=v2--v1--v3};
\draw (v5) -- (v1) -- (v3) pic[draw,"$\beta$",angle radius=1.5cm, angle eccentricity=.8] {angle=v3--v1--v5};
\draw[-{Stealth[scale=1.5]}] (v1) -- node[above] {$\nabla Q(x_t)$} (v3);
\draw[-{Stealth[scale=1.5]}] (v4) -- node[left] {$x_t$} (v1);
\draw (v3) circle (1);
\draw ($(v2)+(-.1,0)$) -- ($(v2)+(-.1,.1)$) -- ($(v2)+(-0,.1)$);
\end{tikzpicture}
\caption{Condition on the angles between $x_t$, $\nabla Q(x_t)$ and $\mathbb{E}\Kr_t$, in the region $\lVert x_t \rVert^2 > D$.}
\label{fig:angleconvexity}
\end{figure}
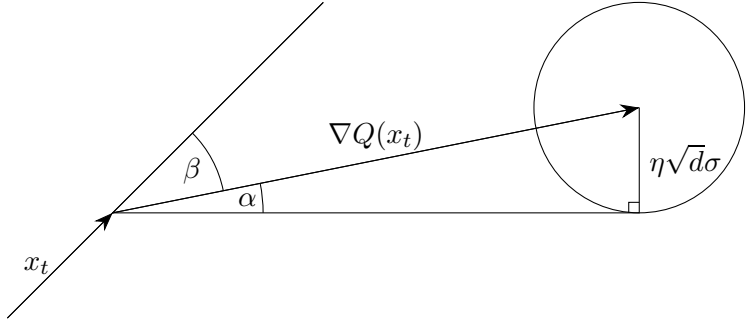




Conditions \emph{(i)} to \emph{(iv)} are the same conditions as in the non-convex convergence analysis in~\cite{bottou1998online}.
Condition \emph{(v)} is a slightly stronger condition than the corresponding one in~\cite{bottou1998online},
and states that, beyond a certain horizon, the cost function $Q$ is ``convex enough'', 
in the sense that the direction of the gradient is sufficiently close to the direction of the parameter vector $x$.
Condition \emph{(iv)}, however, states that the gradient estimator used by the correct workers has to be accurate enough, i.e.,
the local standard deviation should be small relatively to the norm of the gradient. 
Of course, the norm of the gradient tends to zero near, e.g., extremal and saddle points.
Actually, 
the ratio $\eta(n,f)\cdot\sqrt{d}\cdot\sigma / \left\lVert \nabla Q \right\rVert$ 
controls the maximum angle between the gradient $\nabla Q$ and the vector chosen by the Krum function.
In the regions where $\left\lVert \nabla Q \right\rVert < \eta(n,f)\cdot\sqrt{d}\cdot\sigma$, the Byzantine workers 
may take advantage of the noise (measured by $\sigma$) in the gradient estimator $G$ to bias the choice of the parameter server.
Therefore, Proposition~\ref{prop:convergence} is to be interpreted as follows: in the presence of Byzantine workers,
the parameter vector $x_t$ almost surely reaches a basin around points 
where the gradient is small ($\left\lVert \nabla Q \right\rVert \le \eta(n,f)\cdot\sqrt{d}\cdot\sigma$), i.e., 
points where the cost landscape is ``almost flat''.

Note that the convergence analysis is based only on the fact that function $\Kr$ is $(\alpha,f)$-Byzantine resilient.
Due to space limitation, the complete proof of Proposition \ref{prop:convergence} is deferred to the Appendix.

\begin{proof}
  For the sake of simplicity, we write $\Kr_t = \Kr(V^t_1,\dots,V^t_n)$.
  Before proving the main claim of the proposition, 
  we first show that the sequence $x_t$ is almost surely globally confined within the region $\lVert x \rVert^2 \le D$.
  
  \paragraph{\it (Global confinement).}
  Let $u_t = \phi( \lVert x_t \rVert^2)$ where
  \begin{equation*}
    \phi(a) = \left\{ \begin{array}{cc}
     0 &  \text{if } a < D\\
     (a-D)^2 & \text{otherwise}
    \end{array}
    \right.
  \end{equation*}
  Note that
  \begin{equation}
    \phi(b) - \phi(a) \le (b-a)\phi'(a) + (b-a)^2.
    \label{eq:inequphi}
  \end{equation}
  This becomes an equality when $a,b \ge D$. Applying this inequality to $u_{t+1} - u_t$ yields
  \begin{align*}
    u_{t+1} - u_t &\le \left( -2 \gamma_t \langle x_t,\Kr_t\rangle + \gamma_t^2 \lVert \Kr_t \rVert^2 \right) \cdot \phi'( \lVert x_t \rVert^2 ) \\
      &+  4\gamma_t^2 \langle x_t,\Kr_t \rangle^2 - 4\gamma_t^3 \langle x_t, \Kr_t \rangle \lVert \Kr_t \rVert^2 
      + \gamma_t^4 \lVert \Kr_t \rVert^4 \\
      &\le -2\gamma_t \langle x_t, \Kr_t\rangle \phi'(\lVert x_t \rVert^2) + \gamma_t^2 \lVert \Kr_t \rVert^2 \phi'(\lVert x_t \rVert^2) \\
      &+ 4\gamma_t^2 \lVert x_t \rVert^2 \lVert \Kr_t \rVert^2 + 4\gamma_t^3\lVert x_t \rVert \lVert \Kr_t \rVert^3 
      + \gamma_t^4 \lVert \Kr_t \rVert^4.
  \end{align*}
  Let $\mathcal{P}_t$ denote the $\sigma$-algebra encoding all the information up to round $t$. Taking the conditional expectation with respect to $\mathcal{P}_t$
  yields
  \begin{align*}
    \mathbb{E}\left( u_{t+1} - u_t | \mathcal{P}_t\right) &\le -2\gamma_t\langle x_t, \mathbb{E}\Kr_t \rangle + \gamma_t^2 \mathbb{E}\left(\lVert \Kr_t \rVert^2\right) \phi'(\lVert x_t\rVert^2) \\
      &+ 4\gamma_t^2 \lVert x_t \rVert^2 \mathbb{E} \left( \lVert \Kr_t \rVert^2 \right) + 4\gamma_t^3 \lVert x_t \rVert \mathbb{E} \left( \lVert \Kr_t \rVert^3\right) 
      + \gamma_t^4 \mathbb{E} \left(\lVert \Kr_t \rVert^4\right).
  \end{align*}
  Thanks to condition \emph{(ii)} of $(\alpha,f)$-Byzantine resilience, and the assumption on the first four moments of $G$, there exist positive constants $A_0,B_0$ such that
  \begin{equation*}
    \mathbb{E} \left( u_{t+1} - u_t | \mathcal{P}_t \right) \le -2 \gamma_t \langle x_t, \mathbb{E}\Kr_t \rangle \phi'(\lVert x_t \rVert^2) + \gamma_t^2 \left(A_0 + B_0 \lVert x_t \rVert^4\right).
  \end{equation*}
  Thus, there exist positive constant $A,B$ such that
  \begin{equation*}
     \mathbb{E} \left( u_{t+1} - u_t | \mathcal{P}_t \right) \le -2 \gamma_t \langle x_t, \mathbb{E}\Kr_t \rangle \phi'(\lVert x_t \rVert^2) + \gamma_t^2 \left(A + B \cdot u_t \right).
  \end{equation*}
  When $\lVert x_t \rVert^2 < D$, the first term of the right hand side is null because $\phi'(\lVert x_t \rVert^2) = 0$.
  When $\lVert x_t \rVert^2 \ge D$, this first term is negative because (see Figure~\ref{fig:angleconvexity})
  \begin{equation*}
    \langle x_t , \mathbb{E}\Kr_t \rangle \ge \lVert x_t \rVert \cdot \lVert \mathbb{E}\Kr_t \rVert \cdot \cos (\alpha + \beta) 
      > 0.
  \end{equation*}
  Hence
  \begin{equation*}
    \mathbb{E} \left( u_{t+1} - u_t | \mathcal{P}_t \right) \le \gamma_t^2 \left(A + B \cdot u_t \right).
  \end{equation*}
  We define two auxiliary sequences
  \begin{align*}
    \mu_t &= \prod_{i=1}^t \frac{1}{1-\gamma_i^2 B} \xrightarrow[t\rightarrow\infty]{} \mu_{\infty} \\
    u'_t &= \mu_t u_t.
  \end{align*}
  Note that the sequence $\mu_t$ converges because $\sum_t \gamma_t^2 < \infty$. Then
  \begin{equation*}
    \mathbb{E} \left( u'_{t+1} - u'_t | \mathcal{P}_t \right) \le \gamma_t^2 \mu_t A.
  \end{equation*}
  Consider the indicator of the positive variations of the left-hand side
  \begin{equation*}
    \chi_t = \left\{ \begin{array}{cc}
     1 & \text{if }  \mathbb{E} \left( u'_{t+1} - u'_t | \mathcal{P}_t \right) > 0\\
     0 & \text{otherwise}
    \end{array}
    \right.
  \end{equation*}
  Then
  \begin{equation*}
    \mathbb{E} \left( \chi_t \cdot (u'_{t+1} - u'_t)\right) \le \mathbb{E} \left( \chi_t \cdot \mathbb{E} \left( u'_{t+1} - u'_t | \mathcal{P}_t \right)\right) \le \gamma^2_t \mu_t A.
  \end{equation*}
  The right-hand side of the previous inequality is the summand of a convergent series.
  By the quasi-martingale convergence theorem~\cite{MetivierSemiMart}, 
  this shows that the sequence $u'_t$ converges almost surely, 
  which in turn shows that the sequence $u_t$ converges almost surely, $u_t \rightarrow u_{\infty} \ge 0$.
  
  Let us assume that $u_{\infty} > 0$. When $t$ is large enough, 
  this implies that $\lVert x_t \rVert^2$ and $\lVert x_{t+1} \rVert^2$ are greater than $D$.
  Inequality~\ref{eq:inequphi} becomes an equality, which implies that the following infinite sum converges almost surely
  \begin{equation*}
    \sum_{t=1}^{\infty} \gamma_t \langle x_t, \mathbb{E} \Kr_t   \rangle \phi'( \lVert x_t \rVert^2 ) < \infty.
  \end{equation*}
  Note that the sequence $\phi'(\lVert x_t \rVert^2)$ converges to a positive value.
  In the region $\lVert x_t \rVert^2 > D$, we have
  \begin{align*}
    \langle x_t, \mathbb{E}\Kr_t \rangle &\ge \sqrt{D} \cdot \left\lVert \mathbb{E}\Kr_t \right\rVert \cdot \cos(\alpha+\beta) \\
      &\ge \sqrt{D} \cdot \left( \left\lVert \nabla Q(x_t) \right\rVert - \eta(n,f)\cdot\sqrt{d}\cdot \sigma(x_t) \right) \cdot \cos(\alpha + \beta) \\
      &\ge \sqrt{D} \cdot \epsilon \cdot (1-\sin \alpha) \cdot \cos(\alpha+\beta) > 0.
  \end{align*}
  This contradicts the fact that $\sum_{t=1}^{\infty} \gamma_t = \infty$.
  Therefore, the sequence $u_t$ converges to zero.
  This convergence implies that the sequence $\lVert x_t \rVert^2$ is bounded, i.e., 
  the vector $x_t$ is confined in a bounded region containing the origin. 
  As a consequence, any continuous function of $x_t$ is also bounded, 
  such as, e.g., $\lVert x_t \rVert^2$, $\mathbb{E} \left\lVert G(x_t,\xi) \right\rVert^2$ 
  and all the derivatives of the cost function $Q(x_t)$.
  In the sequel, positive constants $K_1,K_2,$ etc\dots are introduced whenever such a bound is used.
  
  \paragraph{\it (Convergence).}
  We proceed to show that the gradient $\nabla Q(x_t)$ converges almost surely to zero.
  We define
  \begin{equation*}
    h_t = Q(x_t).
  \end{equation*}
  Using a first-order Taylor expansion and bounding the second derivative with $K_1$, we obtain
  \begin{equation*}
    \left\lvert h_{t+1} - h_t + 2\gamma_t \langle \Kr_t, \nabla Q(x_t) \rangle \right\rvert \le \gamma_t^2 \lVert \Kr_t \rVert^2 K_1 \text{ a.s.}
  \end{equation*}
  Therefore
  \begin{equation}
    \mathbb{E}\left( h_{t+1} - h_t | \mathcal{P}_t\right)
      \le -2 \gamma_t \langle \mathbb{E} \Kr_t, \nabla Q(x_t) \rangle + \gamma_t^2 \mathbb{E} \left( \lVert \Kr_t \rVert^2 | \mathcal{P}_t \right) K_1.
      \label{eq:ineq2}
  \end{equation}
  By the properties of $(\alpha,f)$-Byzantine resiliency, this implies
  \begin{equation*}
    \mathbb{E} \left( h_{t+1} - h_t | \mathcal{P}_t \right) \le \gamma_t^2 K_2K_1,
  \end{equation*}
  which in turn implies that the positive variations of $h_t$ are also bounded
  \begin{equation*}
    \mathbb{E} \left( \chi_t \cdot \left( h_{t+1} - h_t \right)\right) \le \gamma_t^2 K_2K_1.
  \end{equation*}
  The right-hand side is the summand of a convergent infinite sum. By the quasi-martingale convergence theorem,
  the sequence $h_t$ converges almost surely, $Q(x_t) \rightarrow Q_{\infty}$.
  
  Taking the expectation of Inequality~\ref{eq:ineq2}, and summing on $t=1,\dots,\infty$,
  the convergence of $Q(x_t)$ implies that 
  \begin{equation*}
    \sum_{t=1}^{\infty} \gamma_t \langle \mathbb{E} \Kr_t, \nabla Q(x_t) \rangle < \infty \text{ a.s.}
  \end{equation*}
  We now define 
  \begin{equation*}
    \rho_t = \left\lVert \nabla Q(x_t) \right\rVert^2.
  \end{equation*}
  Using a Taylor expansion, as demonstrated for the variations of $h_t$, we obtain
  \begin{equation*}
    \rho_{t+1} - \rho_t \le -2 \gamma_t \langle \Kr_t, \left(\nabla^2 Q(x_t)\right) \cdot \nabla Q (x_t) \rangle + \gamma_t^2 \left\lVert \Kr_t \right\rVert^2 K_3 \text{ a.s.}
  \end{equation*}
  Taking the conditional expectation, and bounding the second derivatives by $K_4$,
  \begin{equation*}
    \mathbb{E} \left( \rho_{t+1} - \rho_t | \mathcal{P}_t \right) \le 2 \gamma_t \langle \mathbb{E}  \Kr_t, \nabla Q(x_t) \rangle K_4 + \gamma_t^2 K_2K_3.
  \end{equation*}
  The positive expected variations of $\rho_t$ are bounded
  \begin{equation*}
  \mathbb{E} \left(\chi_t\cdot \left(\rho_{t+1} - \rho_t\right)\right) \le 2 \gamma_t \mathbb{E} \langle \mathbb{E}  \Kr_t, \nabla Q(x_t) \rangle K_4 + \gamma_t^2 K_2K_3.
  \end{equation*}
  The two terms on the right-hand side are the summands of convergent infinite series.
  By the quasi-martingale convergence theorem, this shows that $\rho_t$ converges almost surely.
  
  We have
  \begin{align*}
    \langle \mathbb{E} \Kr_t, \nabla Q(x_t) \rangle &\ge \left(\left\lVert \nabla Q(x_t) \right\rVert - \eta(n,f)\cdot\sqrt{d}\cdot\sigma(x_t) \right) \cdot \left\lVert \nabla Q(x_t) \right\rVert \\
     &\ge \underbrace{(1-\sin\alpha)}_{>0} \cdot \rho_t.
  \end{align*}
  This implies that the following infinite series converge almost surely
  \begin{equation*}
    \sum_{t=1}^{\infty} \gamma_t \cdot \rho_t < \infty.
  \end{equation*}
  Since $\rho_t$ converges almost surely, and the series $\sum_{t=1}^{\infty} \gamma_t = \infty$ diverges,
  we conclude that the sequence $\lVert \nabla Q(x_t) \rVert$ converges almost surely to zero.
\end{proof}

\section{$m$-Krum}
\label{sec:multikrum}
 
So far, for the sake of simplicity, we defined our Krum function so that it selects only one vector among the $n$ vectors proposed. In fact,
the parameter server could avoid wasting the contribution of the other workers by selecting $m$ vectors instead.
This can be achieved, for instance, by selecting one vector using the Krum function, removing it from the list, 
and iterating this scheme $m-1$ times, as long as $n-m > 2f+2$.
We then define accordingly the $m$-Krum function
\begin{equation*}
  \Kr_m(V_1,\dots,V_n) = \frac{1}{m} \sum_{s=1}^m V_{i^s_*}
\end{equation*}
where the $V_{i^s_*}$'s are the $m$ vectors selected as explained above. Note that the $1$-Krum function is the Krum function defined in Section~\ref{sec:krum}.

\begin{proposition}
  Let $V_1,\dots,V_n$ be any iid random $d$-dimensional vectors, $V_i \sim G$, with $\mathbb{E}G = g$ 
  and $\mathbb{E}\left\lVert G - g \right\rVert^2 = d\sigma^2$. 
  Let $B_1,\dots,B_f$ be any $f$ random vectors, possibly dependent on the $V_i$'s.
  Assume that $2f + 2 < n - m$  and $\eta(n,f)\sqrt{d}\cdot \sigma < \lVert g \rVert$.
  Then, for large enough $n$, the $m$-Krum function $\Kr_m$ is $(\alpha,f)$-Byzantine resilient where $0 \le \alpha < \pi/2$ is defined by
  \begin{equation*}
    \sin \alpha = \frac{ \eta(n,f) \cdot \sqrt{d} \cdot \sigma }{\lVert g \rVert}.
  \end{equation*}
  \label{prop:multikrumresilient}
\end{proposition}

\begin{proof}[Proof (Sketch)]  For large enough $n$, function $\eta(n,f)$ is increasing in the variable $n$. 
In particular, for all $i \in [0, m-1]$, $\eta(n-i,f) \le \eta(n,f)$. 
For each iteration $i$ from $0$ to $m-1$, Proposition \ref{prop:krumresilient} holds (replacing $n$ by $n-i$ since $n-m > 2f+2$) 
and guarantees that each vector among the $m$ vectors falls under the definition of $(\alpha_i, f)$-Byzantine-resilience, 
where $1-\sin(\alpha_i) = 1- \frac{ \eta(n-i,f) \cdot \sqrt{d} \cdot \sigma }{\lVert g \rVert} \geq 1 - \sin \alpha$. 
Then, $\langle \mathbb{E}\Kr_m, g \rangle \ge \frac{1}{m} \sum\limits_{i=0}^{m-1} (1- \sin\alpha_i) \cdot \lVert g \rVert^2  
        \ge (1- \sin\alpha) \cdot \lVert g \rVert^2$. 
The moments are bounded  above by a linear combination of the  upper-bounds on the moments of each of the $m$ vectors. 
\end{proof}


 %
 %




\section{Concluding Remarks} 
\label{sec:conclusion}



At first glance, the Byzantine-resilient machine problem we address in this paper 
can be related to multi-dimensional approximate agreement~\cite{mendes2013multidimensional, herlihy2014computing}. 
Yet, results 
in $d$-dimensional approximate agreement cannot be applied in our context for the following reasons: 
\emph{(a)}~\cite{mendes2013multidimensional, herlihy2014computing}  assume that the set of vectors that can be proposed to an instance of the agreement is bounded so that at least $f+1$ correct workers propose the same vector, which would require a lot of redundant work in our setting; 
and most importantly, \emph{(b)}~\cite{mendes2013multidimensional} requires a local computation by each worker that is in $O(n^d)$. 
While this cost seems reasonable for small dimensions, such as, e.g., mobile robots meeting in a $2D$ or $3D$ space,
it becomes a real issue in the context of machine learning, 
where $d$ may be as high as $160$ billion~\cite{trask2015modeling} (making $d$ a crucial parameter when considering complexities, either for local computations, or for communication rounds).
%
In our case, the complexity of the Krum function  is $O(n^2 \cdot ( d + \log n))$.


A closer  approach to ours has been recently proposed in~\cite{su2016fault, su2016non}. 
In~\cite{su2016fault}, the authors assume a bounded gradient, and their work was an important step towards Byzantine-tolerant machine learning. 
However, their study only deals with parameter vectors of dimension one.  
In~\cite{su2016non} the authors tackle a multi-dimensional situation, using an iterated approximate Byzantine agreement that reaches consensus asymptotically.  This is however only achieved on a finite set of possible environmental states and cannot be used in the continuous context of stochastic gradient descent.

The present work offers many possible extensions. 
First, the question of whether the bound $2f+2 < n$ is tight remains open, so is the question 
on how to tolerate both asynchrony and Byzantine workers.
Second, we have shown that our scheme forces the parameter vector to reach a region 
where the gradient is small relatively to $\eta \cdot \sqrt{d} \cdot \sigma$. 
The question of whether the factor $\eta(n,f) = O(n)$ can be made smaller also remains open. 
Third, the $m$-Krum function iterates the $1$-Krum function $m$ times, multiplying by $m$ the overall computation complexity.
An alternative is to select the first $m$ vectors after computing the score as in the Krum function.
Proving the $(\alpha,f)$-Byzantine-resilience of this alternative remains open. 

\paragraph{Acknowledgment.}

The authors would like to thank to Lê Nguyen Hoang for fruitful discussion and inputs.

\clearpage
\pagenumbering{gobble}


\clearpage
\pagenumbering{roman}

\end{document}